%% file: paper.tex
\documentclass{article}[11pt]

\usepackage{fullpage, amsmath, amsthm, amssymb}
\usepackage{xspace}

\newtheorem{lemma}{Lemma}
\newtheorem{theorem}{Theorem}
\newtheorem{definition}{Definition}
\newtheorem{corollary}{Corollary}

\newcommand{\zero}{\textbf{0}}
\newcommand{\R}{\mathbb{R}}

\newcommand{\E}{\mathbb{E}}

\newcommand{\Minimize}{\textsc{L2Minimize}\xspace}
\newcommand{\Sample}{\textsc{Sample}\xspace}
\newcommand{\SampleMany}{\textsc{SampleMany}\xspace}

\DeclareMathOperator{\herdisctwo}{herdisc_2}
\DeclareMathOperator{\herdisc}{herdisc_\infty}
\DeclareMathOperator{\disctwo}{disc_2}
\DeclareMathOperator{\disc}{disc_\infty}

\DeclareMathOperator{\lindisctwo}{lindisc_2}
\DeclareMathOperator{\tr}{tr}
\DeclareMathOperator{\vol}{vol}

\DeclareMathOperator{\spn}{span}
\DeclareMathOperator{\sign}{sign}
\DeclareMathOperator{\live}{live}
\DeclareMathOperator{\rank}{rank}

\newcommand{\Prin}{\mathcal{S}}

\newtheorem{innercustomthm}{Restatement of Theorem}
\newenvironment{customthm}[1]
  {\renewcommand\theinnercustomthm{#1}\innercustomthm}
  {\endinnercustomthm}

\newtheorem{innercustomcor}{Restatement of Corollary}
\newenvironment{customcor}[1]
  {\renewcommand\theinnercustomcor{#1}\innercustomcor}
  {\endinnercustomcor}

\title{Constructive Discrepancy Minimization with Hereditary L2
  Guarantees}

\author{Kasper Green
  Larsen\thanks{MADALGO. Aarhus
    University. \texttt{larsen@cs.au.dk}. Supported by a Villum Young
    Investigator Grant and an AUFF Starting Grant.}}

\begin{document}

\date{}
\maketitle

\begin{abstract}
In discrepancy
minimization problems, we are given a family of sets
$\mathcal{S} = \{S_1,\dots,S_m\}$, with each $S_i \in \mathcal{S}$ a subset of some universe
$U = \{u_1,\dots,u_n\}$ of $n$ elements. The goal is to find a
coloring $\chi : U \to \{-1,+1\}$ of the elements of $U$ such that each set $S \in
\mathcal{S}$ is colored as evenly as possible. Two classic measures of
discrepancy are $\ell_\infty$-discrepancy defined as $\disc(\mathcal{S},\chi):=\max_{S \in
  \mathcal{S}} | \sum_{u_i \in S} \chi(u_i) |$ and
$\ell_2$-discrepancy defined as $\disctwo(\mathcal{S},\chi):=\sqrt{(1/|\mathcal{S}|)\sum_{S \in
    \mathcal{S}} \left(\sum_{u_i \in S}
    \chi(u_i)\right)^2}$. Breakthrough work by Bansal [FOCS'10] gave a
polynomial time algorithm, based on rounding an SDP, for finding a coloring $\chi$ such that
$\disc(\mathcal{S},\chi) = O(\lg n \cdot \herdisc(\mathcal{S}))$ where
$\herdisc(\mathcal{S})$ is the hereditary $\ell_\infty$-discrepancy of
$\mathcal{S}$. We complement his work by giving a clean and simple $O((m+n)n^2)$ time algorithm for
finding a coloring $\chi$ such $\disctwo(\mathcal{S},\chi) =
O(\sqrt{\lg n} \cdot \herdisctwo(\mathcal{S}))$ where
$\herdisctwo(\mathcal{S})$ is the hereditary $\ell_2$-discrepancy of
$\mathcal{S}$. Interestingly, our algorithm avoids solving an SDP and
instead relies simply on computing eigendecompositions of
matrices. Moreover, we use some of the ideas in our algorithm to speed
up the Edge-Walk algorithm by Lovett and Meka [SICOMP'15] for
non-square matrices.

To prove that our algorithm has the claimed guarantees, we
also prove new inequalities relating both $\herdisc$ and $\herdisctwo$ to
the eigenvalues of the incidence matrix corresponding to $\mathcal{S}$. Our
inequalities improve over previous work by Chazelle and Lvov [SCG'00]
and by Matousek, Nikolov and Talwar [SODA'15+SCG'15]. We believe these inequalities are of independent interest as powerful tools for
proving hereditary discrepancy lower bounds. Finally, we also
implement our algorithm and show that it far outperforms random sampling
of colorings in practice. Moreover, the algorithm finishes in a
reasonable amount of time on matrices of sizes up to $10000 \times 10000$.
\end{abstract}

\thispagestyle{empty}
\newpage

\input{intro}

\input{eigenbounds}
 
\input{upper}

\input{lovettmeka}

\input{experiments}

\input{acknowledge}

\bibliographystyle{abbrv}
\bibliography{biblio}

\appendix

\end{document}

%% file: intro.tex
\section{Introduction}
\label{sec:intro}
Combinatorial discrepancy minimization is an important field with
numerous applications in theoretical computer science, see e.g. the
excellent books by Chazelle~\cite{chazelle:discrepancy} and Matousek~\cite{matousek1999geometric}. In discrepancy
minimization problems, we are typically given a family of sets
$\mathcal{S} = \{S_1,\dots,S_m\}$, with each $S_i \in \mathcal{S}$ a subset of some universe
$U = \{u_1,\dots,u_n\}$ of $n$ elements. The goal is to find a
red-blue coloring of the elements of $U$ such that each set $S \in
\mathcal{S}$ is colored as evenly as possible. More formally, if we
define the $m \times n$ incidence matrix $A$ with $a_{i,j} = 1$ if
$u_j \in S_i$ and $a_{i,j}=0$ otherwise, then we seek a coloring $x
\in \{-1,+1\}^n$ minimizing either the $\ell_\infty$-discrepancy
$\disc(A,x) := \|Ax\|_\infty$ or the $\ell_2$-discrepancy
$\disctwo(A,x) = (1/\sqrt{m})\|Ax\|_2$. We say that the
$\ell_\infty$-discrepancy of $A$ is $\disc(A) := \min_{x \in
  \{-1,+1\}^n} \disc(A,x)$ and the $\ell_2$-discrepancy of $A$ is
$\disctwo(A) := \min_{x \in \{-1,+1\}^n} \disctwo(A,x)$.
With this matrix view, it is
clear that discrepancy minimization makes sense also for general
matrices and not just ones arising from set systems.

Much research has been devoted to understanding both the
$\ell_\infty$- and $\ell_2$-discrepancy of various families of
set systems and matrices. In particular set systems corresponding to
incidences between geometric objects such as axis-aligned rectangles
and points have been studied extensively, see
e.g.~\cite{matousekGamma, matousek:half, alexander:half, larsenDisc}. Another
fruitful line of research has focused on general matrices, including
the celebrated ``Six Standard Devitations Suffice'' result by
Spencer~\cite{spencer}, showing that any $n \times n$ matrix with
$|a_{i,j}|\leq 1$ admits a coloring $x \in \{-1,+1\}^n$ such that
$\disc(A,x) = O(\sqrt{n})$. Spencer also generalized this to $m \times
n$ matrices while guaranteeing $\disc(A,x) = O(\sqrt{n \ln(em/n)})$
when $m \geq n$. Finding low discrepancy colorings for set
systems where each element appears in at most $t$ sets (the matrix $A$
has at most $t$ non-zeroes per column, all bounded by $1$ in absolute value) has also
received much attention. Beck and Fiala~\cite{beckfiala} gave a deterministic algorithm that finds a
coloring $x$ with $\disc(A,x) = O(t)$. Banaszczyk~\cite{Banas:gauss} improved this to
$O(\sqrt{t \lg n})$ when $t \geq \lg n$. Determining whether a
discrepancy of $O(\sqrt{t})$ can be achieved remains one of the
biggest open problems in discrepancy minimization.

\paragraph{Constructive Discrepancy Minimization.}
Many of the original results, like Spencer's~\cite{spencer} and Banaszczyk's~\cite{Banas:gauss} were
purely existential and it was not clear whether polynomial time
algorithms finding such colorings were possible. In fact, Charikar et
al.~\cite{Charikar} presented very strong negative results in this
direction. More concretely, they proved that it is NP-hard to even
distinguish whether the $\ell_\infty$- or $\ell_2$-discrepancy of an
$n \times n$ set system is $0$ or $\Omega(\sqrt{n})$. The first major
breakthrough on the upper bound side was due to Bansal~\cite{bansal}, who
amongst others gave a polynomial time algorithm for finding a coloring
matching the bounds by Spencer. Brilliant follow-up work by Lovett and
Meka~\cite{lovett} gave a simple randomized algorithm, Edge-Walk, achieving the
same in $\tilde{O}((n+m)^3)$ running time for $m \times n$ matrices. A deterministic algorithm for Spencer's result was later given
by Levy et al.~\cite{levydeterm}, although with worse running time. A number of
constructive algorithms were also given for the ``sparse'' set system
case, finally resulting in polynomial time
algorithms~\cite{bansalKomlos, BansalBeyondPartial, BansalCure} matching the
existential results by Banaszczyk.

Another very surprising result in Bansal's seminal paper~\cite{bansal} shows
that, given a matrix $A$, one can find in polynomial time a coloring $x$ achieving an
$\ell_\infty$-discrepancy roughly bounded by the \emph{hereditary}
discrepancy of $A$. Hereditary discrepancy is a notion introduced by
Lov{\'{a}}sz et al.~\cite{lovasz} in order to prove discrepancy lower
bounds. The hereditary $\ell_\infty$-discrepancy of a matrix $A$ is
defined $
\herdisc(A) := \max_B \disc(B), 
$
where $B$ ranges over all matrices obtained by removing a subset of
the columns in $A$. In the terminology of set systems, the hereditary
discrepancy is the maximum discrepancy over all set systems obtained
by removing a subset of the elements in the universe. We also have an
analogous definition for hereditary $\ell_2$-discrepancy:
$\herdisctwo(A) := \max_B \disctwo(B)$. Based on rounding an SDP,
Bansal gave a polynomial time algorithm for finding a coloring $x$
achieving $\disc(A,x) = O(\lg n \herdisc(A))$. This is quite
surprising in light of the strong negative results by Charikar et
al.~\cite{Charikar}, since it shows that is is in fact possible to find a low
discrepancy coloring of an arbitrary matrix as long as all its submatrices
have low discrepancy.

\paragraph{Our Results Overview.}
Our main algorithmic result is an $\ell_2$ equivalent of Bansal's algorithm with
hereditary guarantees. More concretely, we give a polynomial time
algorithm for finding a coloring $x$ such that $\disctwo(A,x) =
O(\sqrt{\lg n} \cdot \herdisctwo(A))$. We note that neither our result nor
Bansal's approximately imply the other: In one direction, the coloring $x$
we find might have very low $\ell_2$ discrepancy, but a very large
value of $\|Ax\|_\infty$. In the other direction, $\herdisc(A)$ may be
much larger than $\herdisctwo(A)$, thus Bansal's algorithm does not give any
guarantees wrt. $\herdisctwo(A)$.

Our algorithm takes a very different
approach than Bansal's in the sense that we completely avoid solving
an SDP. Instead, we first prove a number of new inequalities relating
$\herdisctwo(A)$ and $\herdisc(A)$ to the eigenvalues of $A^TA$. Relating
hereditary discrepancy to the eigenvalues of $A^TA$ was also done by Chazelle
and Lvov~\cite{chazelleLvov} and by Matou{\v{s}}ek et
al.~\cite{factNorms}. However the result by Chazelle and Lvov is too
weak for our applications as it degenerates exponentially fast in the
ratio between $m$ and $n$. The result of Matou{\v{s}}ek et
al. could be used, but can only show that we find a coloring such that
$\disctwo(A,x) = O(\lg^{3/2} n \cdot \herdisctwo(A))$. We believe our new inequalities are of independent interest as
strong tools for proving discrepancy lower bounds. 

With these
inequalities established, we design a simple and efficient
algorithm, inspired by Beck and Fiala's~\cite{beckfiala} algorithm for sparse set
systems. Our key idea
is to find a coloring $x$ that is
almost orthogonal to all the eigenvectors of $A^TA$ corresponding to
large eigenvalues. This in turn means
that $\|Ax\|_2$ becomes bounded by $\herdisctwo(A)$.

As an interesting corollary of our technique, we also manage to speed
up the Lovett-Meka algorithm for non-square matrices. Amongst others,
this improves the running time for Spencer's six standard deviations
results from $\tilde{O}((n+m)^3)$ to $\tilde{O}(mn+n^3)$. 

We now proceed to present the previous results for proving lower
bounds on the hereditary discrepancy of matrices in order to set the
stage for presenting our new results.

\paragraph{Previous Hereditary Discrepancy Bounds.}
One of the most useful tools in proving lower bounds for hereditary
discrepancy is the determinant lower bound proved in the original
paper introducing hereditary discrepancy:
\begin{theorem}[Determinant Lower Bound (Lov{\'{a}}sz et al.~\cite{lovasz})]
\label{thm:determinantbound}
For an $m \times n$ real matrix $A$ it holds that
$$
\herdisc(A) \geq \max_k \max_B \frac{1}{2}|\det(B)|^{1/k},
$$
where $k$ ranges over all positive integers up to $\min\{n,m\}$ and
$B$ ranges over all $k \times k$ submatrices of $A$.
\end{theorem}
While it is easier to bound the max determinant of a submatrix $B$
than it is to bound the discrepancy of a matrix directly, it still
requires one to argue that we can find some $B$ where all eigenvalues are
non-zero. Chazelle and Lvov demonstrated how it suffices to bound the $k$'th
largest eigenvalue of a matrix in order to derive hereditary
discrepancy lower bounds: 
\begin{theorem}[Chazelle and Lvov~\cite{chazelleLvov}]
For an $m \times n$ real matrix $A$ with $m \leq n$, let $\lambda_1
\geq \cdots \geq \lambda_n \geq 0$ denote the eigenvalues of $A^TA$. For any integer $k
\leq m$, it holds that
$$
\herdisc(A) \geq \frac{1}{2}18^{-n/k} \sqrt{\lambda_k}.
$$
\end{theorem}
The result of Chazelle and Lvov has two substantial caveats. First,
it requires $m \leq n$. Since we will be using the \emph{partial
  coloring} framework, we will end up with matrices having very few
columns but many rows. This completely rules out using the above
result for analysing our new algorithm. Since $k \leq m$, the
lower bound also goes down exponentially fast in the gap between $m$
and $n$ (we note that Chazelle and Lvov didn't explicitly state that
one needs $k \leq m$,
but since $\rank(A) \leq m$, we have $\lambda_k = 0$ whenever $k > m$).

Chazelle and Lvov used their eigenvalue bound to prove the following trace bound
which has been very useful in the study of set systems corresponding
to incidences between geometric objects:
\begin{theorem}[Trace Bound (Chazelle and Lvov~\cite{chazelleLvov})]
For an $m \times n$ real matrix $A$ with $m \leq n$, let $M = A^TA$. Then:
$$
\herdisc(A) \geq \frac{1}{4}324^{-n \tr M^2 / \tr^2 M} \sqrt{\tr M/n}.
$$
\end{theorem}

Matou{\v{s}}ek et al.~\cite{factNorms} presented an alternative to the
result of Chazelle and Lvov, relating $\herdisc(A)$ and
$\herdisctwo(A)$ to the sum of singular values of $A$, i.e. they
proved:
\begin{theorem}[Matou{\v{s}}ek et al.~\cite{factNorms}]
For an $m \times n$ real matrix $A$, let $\lambda_1 \geq \cdots \geq
\lambda_n \geq 0$ denote the eigenvalues of $A^TA$. Then
$$
\herdisc(A) \geq \herdisctwo(A) = \Omega\left(\frac{1}{\lg n}\sum_{k=1}^n
\sqrt{\frac{\lambda_k}{mn}}\right).
$$
which for all positive integers $k \leq \min\{m,n\}$ implies:
$$
\herdisc(A) \geq \herdisctwo(A) = \Omega\left(\frac{k}{\lg n}
\sqrt{\frac{\lambda_k}{mn}}\right).
$$
\end{theorem}
Comparing the bound to the result of Chazelle and Lvov, we see
that the loss in terms of the ratio between $k$ and $n$ is much
better. However for $k,m$ and $n$ all within a constant factor of each
other, Chazelle and Lvov's bound implies $\herdisc(A) =
\Omega(\sqrt{\lambda_k})$ whereas the bound of Matou{\v{s}}ek et
al. loses a $\lg n$ factor and gives $\herdisc(A) \geq \herdisctwo(A)
= \Omega(\sqrt{\lambda_k}/\lg n)$ (strictly speaking, the bound in
terms of the sum of $\sqrt{\lambda_k}$'s is incomparable, but the
bound only in terms of the $k$'th largest eigenvalue does lose this factor).

\paragraph{Our Results.}
We first give a new inequality relating $\herdisc(A)$ to the
eigenvalues of
$A^TA$, simultaneously improving over the previous bounds by Chazelle
and Lvov, and by Matou{\v{s}}ek et al.:
\begin{theorem}
\label{thm:eigenInfty}
For an $m \times n$ real matrix $A$, let $\lambda_1 \geq \lambda_2
\geq \cdots \geq \lambda_n \geq 0$ denote the eigenvalues of
$A^TA$. For all positive integers $k \leq \min\{n,m\}$, we have 
$$
\herdisc(A) \geq \frac{k}{2e}\sqrt{\frac{\lambda_k}{m n}}.
$$
\end{theorem}
Notice that our lower bound goes down as $k/\sqrt{m n}$ whereas
Chazelle and Lvov's goes down as $18^{-n/k}$ and requires $m \leq
n$. Thus our loss is exponentially better than theirs. Compared to the
bound by Matou{\v{s}}ek et al., we avoid the $\lg n$ loss (at least
compared to the bound of Matou{\v{s}}ek et al. that is only in terms
of the $k$'th largest eigenvalue and not the sum of eigenvalues).

 Re-executing
Chazelle and Lvov's proof of the trace bound with the above lemma in place of theirs
immediately gives a stronger version of the trace bound as well:
\begin{corollary}
\label{cor:trace}
For an $m \times n$ real matrix $A$, let $M = A^TA$. Then:
$$
\herdisc(A) \geq \frac{\tr^2 M}{8e \min\{n,m\} \tr M^2}\sqrt{\frac{\tr M}{
    \max\{m,n\}}}.
$$
\end{corollary}

In establishing lower bounds on $\herdisctwo(A)$ in terms of
eigenvalues, we need to first prove an equivalent of the determinant
lower bound for non-square matrices (and for $\ell_2$-discrepancy
rather than $\ell_\infty$):
\begin{theorem}
\label{thm:l2determinant}
For an $m \times n$ real matrix $A$, we have
$$
\herdisc(A) \geq \herdisctwo(A) \geq \sqrt{\frac{n}{8 \pi e m} } \det(A^TA)^{1/2n}.
$$
\end{theorem}
We remark that proving Theorem~\ref{thm:l2determinant}
for the $\ell_\infty$-case appears as an exercise
in~\cite{matousek1999geometric} and we make no claim that the proof of
Theorem~\ref{thm:l2determinant} requires any new or deep insights (we
suspect that it is folklore, but have not been able to find a
mentioning of the above theorem in the literature). 
We finally arrive at our main result for lower bounding hereditary $\ell_2$-discrepancy:
\begin{corollary}
\label{thm:l2eigen}
For an $m \times n$ real matrix $A$, let $\lambda_1 \geq \lambda_2
\geq \cdots \geq \lambda_n \geq 0$ denote the eigenvalues of
$A^TA$. For all positive integers $k \leq \min\{n,m\}$, we have 
$$
\herdisctwo(A) \geq \frac{k}{e} \sqrt{ \frac{\lambda_k}{8 \pi m n}}.
$$
\end{corollary}

We note that Theorem~\ref{thm:eigenInfty} actually follows (up to
constant factors) from
Corollary~\ref{thm:l2eigen} using the fact that $\herdisc(A) \geq
\herdisctwo(A)$, but we will present separate proofs of the two
theorems since the direct proof of Theorem~\ref{thm:eigenInfty} is
very short and crisp.

The exciting part in having established Corollary~\ref{thm:l2eigen}, is
that it hints the direction for giving an efficient algorithm for
obtaining colorings $x$ with $\disctwo(A,x)$ being bounded by some
function of $\herdisctwo(A)$. More concretely, we give an algorithm
that is based on computing an eigendecomposition of $A^TA$ and using this to
perform partial coloring that is orthogonal to the eigenvectors corresponding
to the largest eigenvalues. Via Corollary~\ref{thm:l2eigen}, this
gives a coloring with hereditary $\ell_2$ guarantees. The precise
guarantees of our algorithm are given in the following:
\begin{theorem}
\label{thm:algo}
There is an $O((m+n)n^2)$ time algorithm that given an $m
\times n$ matrix $A$, computes a coloring $x \in \{-1,+1\}^n$
satisfying $\disctwo(A,x) = O(\sqrt{\lg n} \cdot \herdisctwo(A))$.
\end{theorem}
We implemented our algorithm and performed various experiments to
examine its practical performance. Section~\ref{sec:experiments} shows
that the algorithm far outperforms random sampling a coloring $x \in
\{-1,+1\}^n$. In fact, it far outperforms random sampling, even if we
repeatedly sample vectors for as long time as our algorithm runs and
use the best one sampled. Moreover, the algorithm is efficient enough
that it can be run on $1000 \times 1000$ matrices in less than $10$
seconds and on matrices of sizes up to $10000 \times 10000$ in
about 4 hours on a standard laptop. While it is conceivable that
Bansal's SDP based approach can be modified to give $\ell_2$
guarantees with a polynomial running time, it seems highly unlikely that it can process such large matrices in a reasonable amount of
time. Moreover, our algorithm is much simpler to analyse and
implement.

Finally, we also use one of the ideas in our algorithm to speed up the
Edge-Walk algorithm by Lovett and Meka~\cite{lovett}:
\begin{theorem}
\label{thm:lovett}
The Edge-Walk procedure in~\cite{lovett} can be implemented such
that all $T$ iterations run in a total of $O(Tmn + n^3 + Tn^2)$ time. 
\end{theorem}
Since their algorithm matching Spencer's six standard devitations
suffice result invoke the Edge-Walk procedure a polylogarithmic number
of times with $T$ being polylogarithmic, we conclude that a coloring
achieving $\disc(A,x) = O(\sqrt{n \ln(em/n)})$ can be found in time
$\tilde{O}(mn+n^3)$. See Section~\ref{sec:lovettmeka} for the proof of Theorem~\ref{thm:lovett}.

%% file: eigenbounds.tex
\section{Eigenvalue Bounds for Hereditary Discrepancy}
In this section, we prove new results relating the hereditary
discrepancy of a matrix $A$ to the eigenvalues of $A^TA$. The section
is split in two parts, one studying hereditary
$\ell_\infty$-discrepancy and one studying hereditary $\ell_2$-discrepancy.

\subsection{Hereditary $\ell_\infty$-discrepancy}
Our first result concerns hereditary $\ell_\infty$-discrepancy and is a
strengthening of the previous bound due to Chazelle and Lvov~\cite{chazelleLvov}
(see Section~\ref{sec:intro}). The simplest formulation is the following:

\begin{customthm}{\ref{thm:eigenInfty}}
For an $m \times n$ real matrix $A$, let $\lambda_1 \geq \lambda_2
\geq \cdots \geq \lambda_n \geq 0$ denote the eigenvalues of
$A^TA$. For all positive integers $k \leq \min\{n,m\}$, we have 
$$
\herdisc(A) \geq \frac{k}{2e}\sqrt{\frac{\lambda_k}{m n}}.
$$
\end{customthm}

Theorem~\ref{thm:eigenInfty} is an immediate corollary of the
following slightly more general result:

\begin{theorem}
\label{thm:eigenInftyDetails}
For an $m \times n$ real matrix $A$, let $\lambda_1 \geq \lambda_2
\geq \cdots \geq \lambda_n \geq 0$ denote the eigenvalues of
$A^TA$. For all positive integers $k \leq \min\{n,m\}$, we have 
$$
\herdisc(A) \geq \frac{1}{2}\left(\frac{\prod_{i=1}^k
      \lambda_i}{\binom{n}{k} \binom{m}{k}}\right)^{1/2k} 
$$
\end{theorem}

Theorem~\ref{thm:eigenInfty} follows from
Theorem~\ref{thm:eigenInftyDetails} by using that $\binom{n}{k} \leq
(en/k)^k$ and that $\prod_{i=1}^k
      \lambda_i \geq \lambda_k^k$. Thus our goal is to prove
      Theorem~\ref{thm:eigenInftyDetails}. The first step of our proof
      uses the following linear algebraic fact:

\begin{lemma}
\label{lem:detSubmatrix}
For an $m \times n$ real matrix $A$, let $\lambda_1 \geq \lambda_2
\geq \cdots \geq \lambda_n \geq 0$ denote the eigenvalues of
$A^TA$. For all positive integers $k \leq n$, there exists
an $m \times k$ submatrix $C$ of $A$ such that 
$
\det(C^TC) \geq (\prod_{i=1}^k \lambda_i)/\binom{n}{k}
$.
\end{lemma}

\begin{proof}
The $k$'th symmetric function of $\lambda_1,\dots,\lambda_n$ is
defined as (see e.g. the textbook~\cite{MeyerBook} p. 494):
$
s_k = \sum_{1 \leq i_1 < \cdots < i_k \leq n} \lambda_{i_1} \cdots \lambda_{i_k}
$.
Since all $\lambda_i$ are non-negative, we have $s_k \geq
\prod_{i=1}^k \lambda_i$. If we let $\Prin_k(A^TA)$ denote the set of
all $k \times k$ principal submatrices of $A^TA$, then it also holds
that (see e.g. the textbook~\cite{MeyerBook} p. 494):
$
s_k = \sum_{B \in \Prin_k(A^TA)} \det(B)
$.
Since $|\Prin_k(A^TA)| = \binom{n}{k}$ there must be a $B \in \Prin_k(A^TA)$ for which
$\det(B) \geq \left(\prod_{i=1}^k \lambda_i
\right)/\binom{n}{k}$. Since $B$ is a $k \times k$ principal submatrix
of $A^TA$, it follows that there exists an $m \times k$ submatrix $C$
of $A$ such that $B = C^TC$ and thus $\det(C^TC) \geq \left(\prod_{i=1}^k \lambda_i \right)/\binom{n}{k}$. 
\end{proof}

With Lemma~\ref{lem:detSubmatrix} established, we are ready to present
the proof of Theorem~\ref{thm:eigenInftyDetails}:

\begin{proof}[Proof of Theorem~\ref{thm:eigenInftyDetails}]
Let $A$ be a real $m \times n$ matrix and let $\lambda_1 \geq \cdots
\geq \lambda_n \geq 0$ denote the eigenvalues of $A^TA$. From
Lemma~\ref{lem:detSubmatrix}, it follows that for every $k \leq
n$, there is an $m \times k$ submatrix $C$ of $A$ such that
$
\det(C^TC) \geq (\prod_{i=1}^k \lambda_i)/\binom{n}{k}
$.
If we also have $k \leq m$, we can let $\Prin_k(C)$ denote the set of all $k \times k$ principal
submatrices of $C$ and use the Cauchy-Binet formula to conclude that:
$
\det(C^TC) = \sum_{D \in \Prin_k(C)} \det(D)^2
$.
But $\Prin_k(C) \subseteq \Prin_k(A)$ hence there must exist a $k \times k$ matrix $D \in \Prin_k(A)$ such that 
\begin{eqnarray*}
\det(D)^2 \geq \frac{\det(C^TC) }{ |\Prin_k(C)|} \geq \frac{\prod_{i=1}^k \lambda_i}{\binom{n}{k} \binom{m}{k}} \Rightarrow
|\det(D)| \geq \sqrt{\frac{\prod_{i=1}^k \lambda_i}{\binom{n}{k} \binom{m}{k}}}.
\end{eqnarray*}
It follows from the determinant lower bound for hereditary discrepancy
(Theorem~\ref{thm:determinantbound}) that
$$
\herdisc(A) \geq \frac{1}{2}|\det(D)|^{1/k} \geq \frac{1}{2}\left(\frac{\prod_{i=1}^k
      \lambda_i}{\binom{n}{k} \binom{m}{k}}\right)^{1/2k}.
$$
\end{proof}

Having established a stronger connection between eigenvalues and
hereditary discrepancy than the one given by Chazelle and
Lvov~\cite{chazelleLvov}, we can also re-execute their proof of the
trace bound and obtain the following strengthening:

\begin{customcor}{\ref{cor:trace}}
For an $m \times n$ real matrix $A$, let $M = A^TA$. Then:
$$
\herdisc(A) \geq \frac{\tr^2 M}{8e \min\{n,m\} \tr M^2}\sqrt{\frac{\tr M}{
    \max\{m,n\}}}.
$$
\end{customcor}

\begin{proof}
Let $\lambda_1 \geq \cdots \geq \lambda_n \geq 0$ denote the
eigenvalues of $M$. Chazelle and Lvov~\cite{chazelleLvov} proved that if we choose $k =
\tr^2 M/(2 \tr M^2)$ then $\lambda_k \geq \tr M /(4n)$. Examining
their proof, one can in fact strengthen it slightly to $\lambda_k \geq
\tr M/(4 \min\{m,n\})$ (their proof of (\cite{chazelleLvov} Lemma 2.4) considers a uniform
random eigenvalue $\lambda$ amongst $\lambda_1, \dots, \lambda_n$ and
uses that $\tr M = n \E[\lambda]$. However, one needs only $\lambda$ to
be uniform random amongst the non-zero eigenvalues and there are at
most $\min\{m,n\}$ such eigenvalues yielding $\tr M = \min\{n,m\} \E[\lambda]$). Inserting these bounds in
Theorem~\ref{thm:eigenInfty} gives us
$$
\herdisc(A) \geq \frac{\tr^2 M}{8e \tr M^2}\sqrt{\frac{\tr M}{m n
    \min\{m,n\}}} = \frac{\tr^2 M}{8e \min\{n,m\} \tr M^2}\sqrt{\frac{\tr M}{
    \max\{m,n\}}}.
$$
\end{proof}

\subsection{Hereditary $\ell_2$-discrepancy}
This section proves the following determinant result for hereditary
$\ell_2$-discrepancy of $m \times n$ matrices:

\begin{customthm}{\ref{thm:l2determinant}}
For an $m \times n$ real matrix $A$ with $\det(A^TA) \neq 0$, we have
$$
\herdisc(A) \geq \herdisctwo(A) \geq \sqrt{\frac{n m}{8 \pi e} } \det(A^TA)^{1/2n}.
$$
\end{customthm}

The fact $\herdisc(A) \geq \herdisctwo(A)$ is true for all $A$, thus
the difficulty in proving Theorem~\ref{thm:l2determinant} lies in
establishing that
$
\herdisctwo(A) \geq \sqrt{n m/(8 \pi e)} \det(A^TA)^{1/2n}
$.
Our proof uses many of the ideas from the proof of the
determinant lower bound (Theorem~\ref{thm:determinantbound})
in~\cite{lovasz}. We start by introducing the linear discrepancy in
the $\ell_2$ setting and summarize known relations between linear
discrepancy and hereditary discrepancy.

\begin{definition}
Let $A$ be an $m \times n$ real matrix. Then its linear $\ell_2$-discrepancy
is defined as:
$$
\lindisctwo(A) := \max_{c \in [-1,+1]} \min_{x \in \{-1,+1\}^n} \frac{1}{\sqrt{m}}\|A(x-c)\|_2.
$$
\end{definition}

The linear $\ell_2$-discrepancy has a clean geometric
interpretation (this is a direct translation of the similar
interpretation of linear $\ell_\infty$-discrepancy given
e.g. in~\cite{lovasz, matousek1999geometric}). For an $m \times n$ real matrix $A$, let:
$
U_A := \{x : \|Ax\|_2 \leq \sqrt{m}\}
$.
For $t>0$, place $2^n$ translated copies $U_1,\dots,U_{2^n}$ of $tU_A$ such that
there is one copy centered at each point in $\{-1,+1\}^n$. Then $\lindisctwo(A)$ is the least number $t$ for which the sets $U_j$
  cover all of $[-1,+1]^n$.

We will need the following relationship between the hereditary and
linear discrepancy:

\begin{lemma}[Lov{\'{a}}sz et al.~\cite{lovasz}]
\label{lem:linher}
For all $m \times n$ real matrices $A$, it holds that $\lindisctwo(A) \leq 2 \herdisctwo(A)$.
\end{lemma}

We remark that~\cite{lovasz} proved Lemma~\ref{lem:linher} only for the
$\ell_\infty$-discrepancy, but their proof only uses the fact that $\{ x
: \|Ax\|_\infty \leq 1\}$
is centrally symmetric and convex (see~\cite{lovasz} Lemma 1). The
same is true for the $U_A$ defined above.

In light of Lemma~\ref{lem:linher}, we set out to lower bound the
linear discrepancy of an $m \times n$ matrix $A$ in terms of
$\det(A^TA)$. We will prove the following lemma using an adaptation of
the ideas in~\cite{lovasz} (we have not been able to find a proof of this
result elsewhere, but remark that the case
of $m = n$ should follow by adapting the proof in~\cite{lovasz}):

\begin{lemma}
\label{lem:lineardeterm}
Let $A$ be an $m \times n$ real matrix with $\det(A^TA) \neq 0$. Then
$
\lindisctwo(A) \geq \sqrt{n/(2 \pi e m)} \det(A^TA)^{1/2n}
$.
\end{lemma}

\begin{proof}
From the geometric interpretation given earlier, we know that if we
place a copy of
$\lindisctwo(A) U_A$ on each point in $\{-1,+1\}^n$, then they cover
all of $[-1,1]^n$ hence $\vol(\lindisctwo(A) U_A) \geq
\vol([-1,1]^n)/2^n = 1$. But
\begin{eqnarray*}
\vol(\lindisctwo(A) U_A) &=& (\lindisctwo(A))^n \vol(U_A) \\
&=&  (\lindisctwo(A))^n\vol(\{x : \|Ax\|_2 \leq \sqrt{m}\}) \\
&=& (\lindisctwo(A))^n\vol(\{x : x^TA^TAx \leq m\}).
\end{eqnarray*}
Observe now that $\{x : x^TA^TAx \leq m\} = \{x : x^T(m^{-1}A^TA)x \leq 1\}$ is an ellipsoid. It is
well-known that the volume of such an ellipsoid equals
$v_n/\sqrt{\det(m^{-1}A^TA)} = v_n/\sqrt{m^{-n} \det(A^TA)}$ where $v_n$ is the volume of the
$n$-dimensional $\ell_2$ unit ball. Since 
$
v_n = \pi^{n/2}/\Gamma(n/2+1) \leq (2 \pi e/n)^{n/2}
$, we conclude:
\begin{eqnarray*}
1 &\leq&  \frac{(\lindisctwo(A))^n v_n}{\sqrt{m^{-n}\det(A^TA)}} \Rightarrow \\
1 &\leq&  (\lindisctwo(A))^n \left(\frac{2 \pi e m}{ n}\right)^{n/2}
\frac{1}{\sqrt{\det(A^TA)}} \Rightarrow \\
\lindisctwo(A) &\geq& \sqrt{\frac{n}{2 \pi e m}} \det(A^TA)^{1/2n}.
\end{eqnarray*}
\end{proof}

Combining Lemma~\ref{lem:linher} and Lemma~\ref{lem:lineardeterm}
proves Theorem~\ref{thm:l2determinant}.

Having establishes Theorem~\ref{thm:l2determinant}, we are ready to
prove our last result on hereditary $\ell_2$-discrepancy:

\begin{customcor}{\ref{thm:l2eigen}}
For an $m \times n$ real matrix $A$, let $\lambda_1 \geq \lambda_2
\geq \cdots \geq \lambda_n \geq 0$ denote the eigenvalues of
$A^TA$. For all positive integers $k \leq \min\{n,m\}$, we have 
$
\herdisctwo(A) \geq (k/e) \sqrt{\lambda_k/(8 \pi m n)}
$.
\end{customcor}

\begin{proof}
Let $A$ be an $m \times n$ real matrix and let $\lambda_1 \geq \cdots
\geq \lambda_n \geq 0$ be the eigenvalues of $A^TA$. From
Lemma~\ref{lem:detSubmatrix}, we know that for all $k \leq n$, there
is an $m \times k$ submatrix $C$ of $A$ such that
$
\det(C^TC) \geq (\prod_{i=1}^k \lambda_i)/\binom{n}{k} \geq
(k \lambda_k/(e n))^k
$.
From Theorem~\ref{thm:l2determinant}, we get that
$\herdisctwo(C) \geq  \sqrt{k/(8 \pi e m)} \det(C^TC)^{1/2k}
\geq (k/e) \sqrt{\lambda_k/(8 \pi m n)}$.
Since $C$ is obtained from $A$ by deleting a subset of the columns, it
follows that $\herdisctwo(A) \geq \herdisctwo(C)$, completing the proof.
\end{proof}

%% file: upper.tex
\section{Discrepancy Minimization with Hereditary $\ell_2$ Guarantees}
\label{sec:algo}
This section gives our new algorithm for discrepancy minimization. The
goal is to prove the following:

\begin{customthm}{\ref{thm:algo}}
There is an $O((m+n)n^2)$ time algorithm that given an $m
\times n$ matrix $A$, computes a coloring $x \in \{-1,+1\}^n$
satisfying $\disctwo(A,x) = O(\sqrt{\lg n} \cdot \herdisctwo(A))$.
\end{customthm}

Our algorithm follows the same overall approach as several previous algorithms. The general setup is that we first give a procedure for
partial coloring. This procedure takes a matrix $A$ and a partial
coloring $x \in [-1, + 1]^n$. We say that coordinates $i$ of $x$ such
that $|x_i| < 1$ are \emph{live}. If there are $k$ live
coordinates prior to calling the partial coloring method, then upon
termination we get a new vector $\gamma$ such that the number of live
coordinates in $\hat{x} = x + \gamma$ is no more than $k/2$. At the
same time, all coordinates of $\hat{x}$ are bounded by $1$ in absolute
value and $\|A\hat{x}\|_2$ is not much larger than $\|A x\|_2$.

We start by presenting the partial coloring algorithm and then show
how to use it to get the final coloring.

\subsection{Partial Coloring}
In this section, we present our partial coloring algorithm. The algorithm takes as input an $m \times n$ matrix $A$ and a vector $x \in [-1 ,+1]^n$. We think of
the vector $x$ as a partial coloring. We call a coordinate $x_i$ of $x$ \emph{live} if $|x_i| < 1$ and we let $k$ denote the number of
live coordinates in $x$. For ease of notation, we let $\live_x(i)$ denote the index of the $i$'th live coordinate in $x$ and we define $\oplus_x : \R^n \times \R^k \to \R^n$ as the function such that $a \oplus_x b$ for $a \in \R^n$ and $b \in \R^k$, is the vector obtained from $a$ by adding the $i$'th coordinate of $b$ to the coordinate of index $\live_x(i)$ in $a$ (where $\live_x(i)$ refers to the $i$'th live coordinate in $x$).

Upon termination, the algorithm returns another vector $\gamma \in \R^k$. If we let $\hat{x} = x \oplus_x \gamma$ be the vector in $\R^n$ obtained from $x$ by adding $\gamma_i$ to $x_{\live_x(i)}$, then the partial coloring algorithm guarantees the following:
\begin{enumerate}
\item There are at most $k/2$ live coordinates in $\hat{x}$.
\item For all $i$, we have $|\hat{x}_i| \leq 1$.
\item $\|A\hat{x}\|_2^2-\|Ax\|_2^2 = O(m (\herdisctwo(A))^2)$.
\end{enumerate}
Thus upon termination, the new vector $\hat{x}$ has half as many live coordinates, and the discprenacy did not increase by much. In particular the change is related to the
hereditary $\ell_2$-discrepancy of $A$.

The main idea in our algorithm is to use the connection between
eigenvalues and hereditary $\ell_2$-discrepancy that we proved in
Corollary~\ref{thm:l2eigen}. Our algorithm proceeds in iterations, where in each step it finds a vector $v$ and adds it to $\gamma$. The way we choose $v$ is roughly to find the 
eigenvectors of $A^TA$ and then pick $v$ orthogonal to the
eigenvectors corresponding to the largest eigenvalues. This bounds the
difference $\|A(x\oplus_x (\gamma+v))\|_2-\|A(x \oplus_x \gamma)\|_2$ in terms of the eigenvalues and thus hereditary $\ell_2$-discrepancy. At the same time, we use the ideas by Beck and Fiala (and many later papers) where we include constraints forcing $v$ orthogonal to $e_i$
for every coordinate $i$ that is not live. The algorithm is as follows:

\paragraph{PartialColor($A$, $x$):}
\begin{enumerate}
\item Let $k$ denote the number of live coordinates in $x$ and let $C$ denote the $m \times k$ matrix obtained from $A$ by
  deleting all columns corresponding to coordinates that are not live.
\item Initialize $\gamma = \textbf{0} \in \R^k$. 
\item Compute an eigendecomposition of $C^TC$ to obtain the
  eigenvalues $\lambda_1 \geq \cdots \geq \lambda_k \geq 0$ and corresponding eigenvectors $\mu_1,\dots,\mu_k$.
\item While \textbf{True}:
\begin{enumerate}
\item Compute the set $S$ of coordinates $i$ such that $|\gamma_i + x_{\live_x(i)}|=1$. If $|S| \geq k/2$, \textbf{return} $\gamma$.
\item Find a unit vector $v$ orthogonal to all $e_j$ with $j \in S$ and to all $\mu_i$ with $i \leq k/4$.
\item Let $\sigma = -\sign(\langle Ax, A(\textbf{0} \oplus_x v)\rangle)$. Compute the largest $\beta>0$ such that all coordinates of $x \oplus_x (\gamma+\sigma \beta v)$ are less than or equal to $1$ in absolute value. Update $\gamma \gets \gamma + \sigma \beta v$.
\end{enumerate}
\end{enumerate}

\paragraph{Correctness.}
We prove that the vector $\gamma$ returned by the above \textbf{PartialColor} algorithm satisfies the three claimed properties. First observe that in every iteration of the while loop, we find a vector $v$ that is orthogonal to $e_i$ whenever $|\gamma_i + x_{\live_x(i)}|  = 1$. Hence if $|\gamma_i + x_{\live_x(i)}|$ becomes $1$, it never changes again. Moreover, by maximizing $\beta$ in each iteration, we guarantee that at least one more coordinate satisfies $|\gamma_i + x_{\live_x(i)}|  = 1$ after every iteration. Thus the algorithm terminates after at most $k/2$ iterations of the while loop and no coordinate of $x \oplus_x \gamma$ is larger than $1$ in absolute value. What remains is to bound $\|A(x \oplus_x \gamma)\|_2^2 - \|Ax\|_2^2$.

Let $v^{(i)}$ denote the vector $v$ found during the $i$'th iteration of the while loop. Upon termination, we have that $\gamma=\sigma_1 \beta_1 v^{(1)} + \cdots + \sigma_r \beta_r v^{(r)}$ where $\sigma_i = \sign(\langle Ax, v^{(i)}\rangle)$ and each $v^{(i)}$ is orthogonal to $\mu_1,\dots,\mu_{k/4}$. Thus $\gamma$ is also orthogonal to $\mu_1,\dots,\mu_{k/4}$. We therefore have:
\begin{eqnarray*}
\|A(x \oplus_x \gamma)\|_2^2 &=& \|A(x + (\textbf{0} \oplus_x \gamma))\|_2^2 \\ 
&\leq& \|Ax\|_2^2 + \|A(\textbf{0} \oplus_x \gamma)\|_2^2 + 2 \langle Ax, A(\textbf{0} \oplus_x \gamma) \rangle \\
&=& \|Ax\|_2^2 + \|C \gamma \|_2^2 + 2 \sum_{i=1}^r \langle Ax, A(\textbf{0} \oplus_x \sigma_i \beta_i v^{(i)}) \rangle \\
&\leq& \|Ax\|_2^2  + \lambda_{k/4} \|\gamma\|_2^2 - 2 \sum_{i=1}^r \sign(\langle Ax,A(\textbf{0} \oplus_x v^{(i)})\rangle)  \langle Ax, A(\textbf{0} \oplus_x \beta_i v^{(i)}) \rangle \\
&=& \|Ax\|_2^2  + \lambda_{k/4} \|\gamma\|_2^2 - 2 \sum_{i=1}^r \sign(\langle Ax,A(\textbf{0} \oplus_x v^{(i)})\rangle) ^2 |\langle Ax, A(\textbf{0} \oplus_x \beta_i v^{(i)}) \rangle| \\
&\leq& \|Ax\|_2^2  + \|\gamma\|_\infty^2 k \lambda_{k/4}-0\\
&\leq& \|Ax\|_2^2  + 4 k \lambda_{k/4}.
\end{eqnarray*}
We would like to use Corollary~\ref{thm:l2eigen} to relate $k
\lambda_{k/4}$ to the hereditary discrepancy of $A$. Since $C$ is an $m \times k$ submatrix of $A$, we have $\herdisctwo(A) \geq \herdisctwo(C)$. Using Corollary~\ref{thm:l2eigen} we have $\herdisctwo(C)  \geq (k/4e) \sqrt{\lambda_{k/4}/mk} = (1/4e)\sqrt{k \lambda_{k/4}/(8 \pi) m}$. Hence we conclude that
$$
\|A\hat{x}\|_2^2-\|Ax\|_2^2 \leq 128 e^2\pi m (\herdisctwo(A))^2 = O(m (\herdisctwo(A))^2).
$$

\paragraph{Running Time.}
Step 1. of \textbf{PartialColor} takes $O(mk)$ time and step 2. takes $O(k)$. Step 3. takes $O(mk^2)$ time to compute $C^TC$ (can be improved via fast matrix multiplication) and $O(k^3)$ time to compute the eigendecomposition. As argued above, each iteration of the while loop increases the size of $S$ by at least one. Hence there are no more than $k/2$ iterations of the loop. Computing $S$ in step (a) takes $O(k)$ time. Finding the unit vector $v$ in step (b) can be done in $O(k^2)$ time as follows: Whenever adding a coordinate $i$ to $S$, use Gram-Schmidt to compute the normalized (unit-norm) projection $\hat{e}_i$ of $e_i$ onto the orthogonal complement of $\mu_1,\dots,\mu_{k/4}$ and all previous vectors $\hat{e}_j$. This takes $O(k^2)$ time per $i$. To find $v$, sample a uniform random unit vector in $\R^k$ and run Gram-Schmidt to compute its projection onto the orthogonal complement of $\hat{e}_j$ for $j \in S$ and $\mu_1,\dots,\mu_{k/4}$. The expected length of the projection is $\Omega(1)$ and we can scale it to unit length afterwards. This gives the desired vector. The Gram-Schmidt step takes $O(k^2)$ time. Computing $A(\textbf{0} \oplus_x v)$ in step (c) takes $O(mk)$ time and computing $Ax$ can be done outside the while loop in $O(mn)$ time. The inner product takes $O(m)$ time to compute. Computing $\beta$ and adding $\sigma \beta v$ to $\gamma$ takes $O(k)$ time. Overall, the \textbf{PartialColor} algorithm takes $O(mn + mk^2 + k^3)$ time. If $Ax$ is given as argument to the algorithm, the time is further reduced to $O((m+k)k^2)$.

\subsection{The Final Algorithm}
Now that we have the \textbf{PartialColor} algorithm, getting to a
low discrepancy coloring is straight forward. Given an $m \times n$
matrix $A$, we initialize $x \gets \zero$. We then repeatedly invoke
\textbf{PartialColor}($A$, $x$). Each call returns a vector
$\gamma$. We update $x
\gets x + \gamma$ and continue. We stop
once there are no live coordinates in $x$, i.e. all coordinates
satisfy $|x_i| = 1$.

In each iteration, the number of live coordinates of $i$ decreases by
at least a factor two, and thus we are done after at most $\lg n$
iterations. This means that the final vector $x$ satisfies
\begin{eqnarray*}
\|Ax\|_2^2 &\leq& \lg n \cdot O(m (\herdisctwo(A))^2)
\Rightarrow \\
\|Ax\|_2 &=& O(\sqrt{m \lg n} \cdot \herdisctwo(A)) \Rightarrow \\
\disctwo(A,x) &=& O(\sqrt{\lg n} \cdot \herdisctwo(A)).
\end{eqnarray*}
For the running time, observe that after each call to \textbf{PartialColor}, we can compute $A(x+\gamma)$ from $Ax$ in $O(mk)$ time. Thus we can provide $Ax$ as argument to \textbf{PartialColor} and thereby reduce its running time to $O((m+k)k^2)$. Since $k$ halves in each iteration, we get a running time of
$$
O\left(\sum_{i=1}^{\lg n} (m + n/2^i)(n/2^i)^2\right) = O((m+n)n^2).
$$
This concludes the proof of Theorem~\ref{thm:algo}.

%% file: lovettmeka.tex
\section{Faster Walking on the Edge}
\label{sec:lovettmeka}
In this section, we discuss how our ideas from Section~\ref{sec:algo} may be used to speed up the $\ell_\infty$ discrepancy minimization algorithm Edge-Walk by Lovett and Meka~\cite{lovett}. We have restated the theorem here for convenience:
\begin{customthm}{\ref{thm:lovett}}
The Edge-Walk procedure in~\cite{lovett} can be implemented such
that all $T$ iterations run in a total of $O(Tmn + n^3 + Tn^2)$ time. 
\end{customthm}
We have shown the full Edge-Walk procedure here:\\\\
\textbf{Edge-Walk:} For $t=1,\dots,T$ do:
\begin{enumerate}
\item Let $\mathcal{C}^{\textrm{var}}_t :=\mathcal{C}^{\textrm{var}}_t(X_{t-1}) = \{i \in [n] : |(X_{t-1})_i| \geq 1-\delta \}$ be the set of variable constraints 'nearly hit' so far.
\item Let $\mathcal{C}^{\textrm{disc}}_t :=\mathcal{C}^{\textrm{disc}}_t(X_{t-1}) = \{j \in [m] : |\langle X_{t-1} - x_0, v_j \rangle| \geq c_j - \delta| \}$ be the set of discrepancy constraints 'nearly hit' so far.
\item Let $\mathcal{V}_t := \mathcal{V}(X_{t-1}) = \{u \in \R^n : u_i = 0 \forall i \in \mathcal{C}_t^{\textrm{var}}, \langle u,v_j \rangle=0 \forall j\in \mathcal{C}^{\textrm{disc}}_t \}$ be the linear subspace orthogonal to the 'nearly hit' variable and discrepancy constraints.
\item Set $X_t := X_{t-1} + \gamma U_t$, where $U_t \sim \mathcal{N}(V_t)$.
\end{enumerate}
We give a brief overview of the idea in the above and refer the reader to~\cite{lovett} for more details. Step 1. finds variables that are no longer live, step 2. finds rows  $j$ of the input matrix where the current coloring almost yields a larger discrepancy than a predefined threshold $c_j$. Step 3. defines the subspace $\mathcal{V}_t$ of vectors that are orthogonal to the rows and variables that are nearly violated and step 4. picks a vector $U_t$ by sampling an $\mathcal{N}(0,1)$ random variable $g_i$ for each vector $v_i$ in an orthonormal basis for $\mathcal{V}_t$ and letting $U_t = \sum_i g_i v_i$.

Step 1. takes $O(n)$ time and step 2. takes $O(nm)$ time, for a total of $O(Tnm)$ over all $T$ iterations. Lovett and Meka charge a total of $O((n+m)^3)$ for steps 3. and 4. by explicitly computing the subspace $\mathcal{V}_t$ and sampling a vector from it.

To speed up the above, we notice that as soon as a variable or discrepancy constraint enters either $\mathcal{C}^{\textrm{var}}_t$ or $\mathcal{C}^{\textrm{disc}}_t $ it never leaves again. When a new constraint enters, it either adds the constraint $\langle u, e_i \rangle = 0$ or $\langle u, v_j \rangle=0$ for some $j$. We maintain an orthonormal basis for the subspace spanned by these $e_i$'s and $v_j$'s. This is done as follows: Assume the current orthonormal basis is $x_1,\dots,x_k$ for some $k$ and we receive the constraint $\langle u , y \rangle = 0$ for some $y \in \R^n$. We then use Gram-Schmidt to compute the projection $y^\bot$ of $y$ onto the orthogonal complement of $\spn(x_1,\dots,x_k)$. If $y^\bot \neq 0$, we normalize it and add it to $x_1,\dots,x_k$. This takes $O(nk) = O(n^2)$ time since $|\mathcal{C}^{\textrm{var}}_t| \leq n$ and the Edge-Walk procedure is allowed to fail if $|\mathcal{C}^{\textrm{disc}}_t| \geq n$ (see the analysis in Lovett and Meka~\cite{lovett}). The orthogonalization step just described is performed at most $2n$ times by the same argument, thus the total time for adding constraints is $O(n^3)$. Finally we need to argue how we can sample $U_t$. This is done simply by sampling a vector $y$ in $\R^n$ with each coordinate $\mathcal{N}(0,1)$ distributed. We then compute the projection $y^\bot$ of $y$ onto the orthogonal complement of $\spn(x_1,\dots,x_k)$ using Gram-Schmidt. The resulting vector $y^\bot$ is precisely $\mathcal{N}(V_t)$ distributed. This took $O(n^2)$ time and is performed at most $T$ times. The total running time is hence $O(Tmn+n^3 + Tn^2)$.

%% file: experiments.tex
\section{Experiments}
\label{sec:experiments}
In this section, we present a number of experiments to test the practical performance of our hereditary $\ell_2$ discrepancy minimization algorithm. We denote the algorithm by \Minimize in the following.  We compare it to two base line algorithms \Sample and \SampleMany. \Sample simply picks a uniform random $\{-1,+1\}$ vector as its coloring. \SampleMany repeatedly samples a uniform random $\{-1,+1\}$ vector and runs for the same amount of time as \Minimize. It returns the best vector found within the time limit.

The algorithms were implemented in Python, using NumPy and SciPy for linear algebra operations. All tests were run on a MacBook Pro (15-inch, Late 2013) running macOS Sierra 10.13.3. The machine has a 2 GHz Intel Core i7 and 8GB DDR3 RAM.

We tested the algorithms on three different classes of matrices: 
\begin{itemize}
\item \textbf{Uniform} matrices: Each coordinate is uniform random and independently chosen among $-1$ and $+1$.
\item \textbf{2D Corner} matrices: Obtained by sampling two sets $P=\{p_1,\dots,p_n\}$ and $Q = \{q_1,\dots,q_m\}$ of $n$ and $m$ points in the plane, respectively. The points are sampled uniformly in the $[0,1] \times [0,1]$ unit square. The resulting matrix has one column per point $p_j \in P$ and one row per point $q_i \in Q$. The entry $(i,j)$ is 1 if $p_j$ is \emph{dominated} by $q_i$, i.e. $q_i.x > p_j.x$ and $q_i.y > p_j.y$ and it is $0$ otherwise. Such matrices are known to have hereditary $\ell_2$-discrepancy $O(\lg^2n)$~\cite{larsenDisc} (we are not aware of better upper bounds).
\item \textbf{2D Halfspace} matrices: Obtained by sampling a set $P = \{p_1,\dots,p_n\}$ of $n$ points in the unit square $[0,1] \times [0,1]$, and a set $Q$ of $m$ halfspace. Each halfspace in $Q$ is sampled by picking one point $a$ uniformly on either the left boundary of the unit square or on the top boundary, and another point $b$ uniformly on either the right boundary or the bottom boundary of the unit square. The halfspace is then chosen uniformly to be either everything above the line through $a,b$ or everything below it. The resulting matrix has one column per point $p_j \in P$ and one row per halfspace $h_i \in Q$. The entry $(i,j)$ is $1$ if $p_j$ is in the halfspace $h_i$ and it is $0$ otherwise. Such matrices are known to have hereditary $\ell_2$-discrepancy $O(n^{1/4})$~\cite{matousek:half}.
\end{itemize}

Each test is run $10$ times and the average $\ell_2$ discrepancy and average runtime is reported. The running times of the algorithms varied exclusively with the matrix size and not the type of matrix, thus we only show one time column which is representative of all types of matrices. The results are shown in Table~\ref{tbl}.
\begin{center}
\begin{table}[h]

\centering
\begin{tabular}{|c|c|r|r|r|r|}
\hline
Algorithm & Matrix Size & Disc Uniform & Disc 2D Corner & Disc 2D Halfspace & Time (s)\\
\hline
\Minimize & $200 \times 200$ & 7.2 & 1.8 & 1.6 & $< 1$\\
\Sample & $200 \times 200$ & 13.8 & 7.6 & 11.0 & $<1$\\
\SampleMany & $200 \times 200$ & 11.6 & 2.3 & 2.7 & $<1$\\
\hline
\Minimize & $1000 \times 1000$ & 15.7 & 1.9 & 2.3 & 9\\
\Sample &  $1000 \times 1000$ & 31.6 & 16.0 & 18.3 & $<1$\\
\SampleMany &  $1000 \times 1000$ & 28.9 & 4.9 & 5.5 & 9\\
\hline
\Minimize & $4000 \times 4000$ & 31.0 & 2.1 & 2.6 & 717\\
\Sample &  $4000 \times 4000$ & 63.1 & 21.0 & 34.0 & $<1$\\
\SampleMany &  $4000 \times 4000$ & 60.3 & 9.5 & 10.7 & 717\\
\hline
\Minimize & $10000 \times 10000$ & 48.3 & 2.1 & 3.1 & 15260\\
\Sample &  $10000 \times 10000$ & 99.9 & 51.4 & 96.8 & $<1$\\
\SampleMany &  $10000 \times 10000$ & 96.8 & 14.2 & 15.6 & 15260\\
\hline
\Minimize & $10000 \times 2000$ & 35.9 & 2.1 & 2.7 & 535\\
\Sample &  $10000 \times 2000$ & 44.7 & 20.6 & 24.1 & $<1$\\
\SampleMany &  $10000 \times 2000$ & 43.4 & 6.7 & 8.0 & 535\\
\hline
\Minimize & $2000 \times 10000$ & 21.4 & 1.8 & 2.0 & 5809\\
\Sample &  $2000 \times 10000$ & 99.9 & 40.8 & 70.8 & $<1$\\
\SampleMany &  $2000 \times 10000$ & 92.2 & 13.8 & 16.4 & 5809\\
\hline
\end{tabular}
\caption{Results of experiments with our \Minimize algorithm. The Matrix Size column gives the size $m \times n$ of the input matrix. The Disc columns shows $\disctwo(A,x) = \|Ax\|_2/\sqrt{m}$ for the coloring $x$ found by the algorithm on the given type of matrix. Time is measured in seconds. Each entry is the average of 10 executions.}
\label{tbl}
\end{table}
\end{center}
The table clearly shows that \Minimize gives superior colorings for all types of matrices and all sizes. The tendency is particularly clear on the structured matrices \textbf{2D Corner} and \textbf{2D Halfspace} where the coloring found by \Minimize on $10000 \times 10000$ matrices is a factor 25-30 smaller than a single round of random sampling (\Sample) and a factor 5-7 better than random sampling for as long time as \Minimize runs (\SampleMany).

The $O((m+n)n^2)$ running time makes the algorithm practical up to matrices of size about $10000 \times 10000$, at which point the algorithm runs for $15260$ seconds $\approx$ $4$ hours and $15$ minutes. 

%% file: acknowledge.tex
\section{Acknowledgment}
The author wishes to thank Nikhil Bansal for useful discussions and
pointers to relevant literature. The author also thanks an
anonymous STOC reviewer for comments that simplified the
\textbf{PartialColor} algorithm.